\documentclass[a4paper,superscriptaddress,twocolumn,pra,showkeys]{revtex4-1} 

\usepackage{graphicx} 
\usepackage{epstopdf}
\usepackage{color} 
\usepackage{amsmath} 		
\usepackage{amssymb}  
\usepackage{amsthm}
\usepackage{capt-of} 
\usepackage{natbib}

\newtheorem{theorem}{Theorem}[section]

\newtheorem{proposition}[theorem]{Proposition}

\begin{document}

\title{Flag-based Control of Quantum Purity for $n=2$ Systems}

\author{Patrick Rooney}
\email{darraghrooney@gmail.com}
\noaffiliation

\author{Anthony M. Bloch}
\email{abloch@umich.edu}
\affiliation{Department of Mathematics, University of Michigan, Ann Arbor, MI 48109}

\author{C. Rangan}
\email{rangan@uwindsor.ca}
\affiliation{Department of Physics, University of Windsor, ON,
N9B 3P4, Canada}

\begin{abstract}
This paper investigates the fast Hamiltonian control of $n=2$ density operators by continuously varying the flag as one moves away from the completely mixed state. In general, the critical points and zeros of the purity derivative can only be solved analytically in the limit of minimal purity. We derive differential equations that maintain these features as the purity increases. In particular, there is a thread of points in the Bloch ball that locally maximizes the purity derivative, and a corresponding thread that minimizes it. Additionally, we show there is a closed surface of points inside of which the purity derivative is positive, and inside of which is negative. We argue that this approach may be useful in studying higher-dimensional systems.
\end{abstract}


\maketitle

\section{Introduction}

In the last three decades, there has been great interest in controlling quantum systems for the purposes of coherent control of chemical reactions \cite{ShapiroBrumer86}\cite{TannorRice85}, NMR \cite{ErnstetalBook}, and quantum computation \cite{RanganBucksbaum01}\cite{PalaoKosloff02}. One of the key challenges of quantum control is counter-acting the influence of the environment, which causes decoherence (loss of coherence between quantum states) and dissipation (loss of energy) (see \cite{MabuchiKhaneja2005}, \cite{BrifChakrabartiRabitz2010}, \cite{DongPetersen2011} and \cite{AltafiniTicozzi2012} for surveys). If one models an open system as Markovian and time-independent, the dynamics are described by a quantum dynamical semi-group and the Lindblad master equation \cite{Lindblad76}\cite{GoriniKossakowskiSudarshan76}\cite{BreuerPetruccioneBook}. While there is research towards engineering open system dynamics\cite{LloydViola01}\cite{Baconetal01}\cite{Barreiroetal2011}, control functions often appear in the system Hamiltonian, which are only capable of steering within a given unitary orbit \cite{TannorBartana99}\cite{SklarzTannorKhaneja04}\cite{Schirmeretal04}. The motion \emph{between} orbits depends on the Lindblad super-operator. This includes, in particular, variation in the purity $Tr(\rho^2)$, which is constant on any given orbit. This incentivizes  unitary control that is fast relative to the time-scale of the Lindblad dynamics\cite{Khanejaetal02}.

One method of representing open systems is the generalized Bloch representation \cite{Schirmeretal04}\cite{SchirmerWang2010}, which yields an affine differential equation on the vector space of density operators. In this paper we use a different approach in which the structure of the space of density operators is decomposed into the space of unitary orbits, and the orbit manifolds themselves.  If one has sufficiently fast and complete Hamiltonian control for an $n=2$ system, the inter- and intra-orbit dynamics can be turned into a control equation, where the position along the orbit is considered a control variable, and the orbit itself is treated as a state variable \cite{us_nis2_a}\cite{Yuan2013}. Mathematically, the position along the orbit is a \emph{flag}: a nesting of subspaces, which in this case are the eigenspaces of the density operator. Therefore, we refer to this viewpoint as flag-based control. 

The difficulty in this approach is that the orbit is a non-linear manifold, and applying standard control theory to obtain explicit trajectories is non-trivial. Controllability for $n=2$ can be treated analytically, but not in a way that will scale practically to higher dimensions. At the completely mixed state however, the structure of the Lindblad term simplifies significantly regardless of dimension. This paper considers an approach that begins at the completely mixed state, and introduces a continuously varying feedback in the flag, which maintains critical points and zeros of the purity derivative as the purity increases. In this way, one can plot a thread through $\rho$-space, the so-called Bloch ball, that maximizes (at least locally) the time-derivative of $Tr(\rho^2)$. There is a corresponding thread that minimizes the time-derivative (better said, maximizes it in the negative direction). Additionally, a different feedback can be derived that maintains a surface that separates the Bloch ball into purity-increasing and purity-decreasing regions. The purity derivative on this surface vanishes.

In the section II, we decompose the Lindblad master equation into its inter- and intra-orbit components, and interpret the resulting ODE as a control equation. In section III, we derive a feedback equation that maintains critical points as purity varies. We consider the special cases where the feedback equation fails. In section IV, we derive the feedback that maintains zeros. In section V we show several examples that illustrate the effect of the Lindblad super-operator on the Bloch ball.

\section{Preliminaries}

An open quantum system is described by a density operator $\rho$, which is a trace-one positive semi-definite operator on the Hilbert space. If the dissipation is Markovian and time-independent, the density operator obeys the Lindblad equation \cite{Lindblad76}:
\begin{align}
\frac{d}{dt}\rho(t) = [-iH(t), \rho(t)] + \mathcal{L}_D(\rho(t)) \\
\mathcal{L}_D(\rho) := \sum_{m=1}^N L_m\rho L_m^\dagger - \frac{1}{2}\{L_m^\dagger L_m, \rho \} ,
\end{align}
where the braces indicate an anti-commutator, $H(t)$ is the (Hermitian) Hamiltonian, and $\{L_m\}$ are the so-called Lindblad operators. 

For $n=2$, the density operator can be identified with the Bloch vector $\vec{n} \in \mathbb{R}^3$, $|\vec{n}|\le 1$. The identification is:
\begin{align}
\rho = \frac{1}{2}\left(I_2 + \sum_{j=x,y,z}n_j\sigma_j \right),
\end{align}
where $\{\sigma_j\}$ are the Pauli matrices. The Lindblad equation translates to the following ODE (see Appendix A for a derivation): 
\begin{align}
\frac{d\vec{n}}{dt} = \vec{b} + \vec{h}\times\vec{n} +(A - \textrm{tr}(A))\vec{n}.
\end{align}
Here we write $\vec{h}$ and $\vec{l}_m$ to represent the traceless parts of the operators $H$ and $L_m$, expressed in the basis of the Pauli matrices, so that $H = h_0I +\sum_{j=1}^3 h_j\sigma_j$. The system parameters are defined:
\begin{align}
A &:= \frac{1}{2}\sum_m \vec{l}_m\bar{\vec{l}}_m^T + \bar{\vec{l}}_m\vec{l}_m^T\\
\vec{b} &:= i \sum_m \vec{l}_m \times \bar{\vec{l}}_m,
\end{align}
where the bar represents complex conjugate and $T$ matrix transpose. $A$ is a positive semi-definite matrix, so its eigenvalues $a_j$ must be non-negative. Additionally, the vector $\vec{b}$ obeys the inequality (see Appendix B):
\begin{align}
\vec{b}^T A \vec{b} \le 4 \det(A).\label{ineq}
\end{align}

For $\vec{n}\ne 0$, the ODE can be decoupled into its radial and transverse components. If we write $\vec{n}=r\hat{n}$, then $\frac{d\vec{n}}{dt} = \frac{dr}{dt} \hat{n}+r\frac{d\hat{n}}{dt}$, which yields $\frac{dr}{dt} = \hat{n}\cdot\frac{d\vec{n}}{dt}$, as well as $\frac{d\hat{n}}{dt} = -\frac{1}{r}\hat{n}\times(\hat{n}\times\frac{d\vec{n}}{dt})$. Then we have:
\begin{align}
\frac{dr}{dt} &= \vec{b}\cdot\hat{n} + r\left(\hat{n}\cdot A\hat{n} - \textrm{tr}(A)\right) =: f(\hat{n},r) \label{radial}\\
\frac{d\hat{n}}{dt} &= - \frac{1}{r}\vec{b}_\perp + \vec{h}\times \hat{n} - (A\hat{n})_\perp,  \label{transverse}
\end{align}
where the $\perp$ subscript indicates the component perpendicular to $\hat{n}$. 

The behavior at $r=0$ demands attention. Clearly, $\frac{d\hat{n}}{dt}$  can be quite large for small $r$, but for trajectories that pass through $\vec{n}=0$, it is well behaved. At this point, $\frac{d\vec{n}}{dt} = \vec{b}$, which means that shortly before or after, we have $\vec{n} = \vec{b}\hspace{3pt} \delta t$. It follows that $r = |\vec{b}\hspace{3pt}\delta t|$ and $\hat{n} = \textrm{sgn}(\delta t) \hat{b}$. The ODE's then give:
\begin{align}
\frac{dr}{dt} &= \textrm{sgn}(\delta t)|\vec{b}| + O(\delta t) \\
\frac{d\hat{n}}{dt} &= \textrm{sgn}(\delta t) \left( \vec{h} \times \hat{b} - (A\hat{b})_\perp\right), 
\end{align}
which are clearly bounded.

\section{Maximizing and Minimizing Threads}

In a control-theoretic context, we are typically able to choose the Hamiltonian $H(t)$ to some degree. The Hamiltonian appears only in the transverse equation (\ref{transverse}), while the radial component (\ref{radial}) has no explicit Hamiltonian dependence. We are interested in the question of how to steer the transverse component in order to influence the radial.  We will presume that we have fast and complete controllability, \emph{i.e.}, in the absence of Lindblad dissipation, we are able to steer between any two points on an orbit in arbitrarily short time (or at least in a time-scale much shorter than that associated with the Lindblad operators). This means we can consider $\hat{n}$ to be an effective control variable: we can search for desirable $\hat{n}(t)$ and then re-construct $H(t)$ afterwards. 

Given this context, we are interested in optimizing the function $f_r(\hat{n}) := f(\hat{n},r)$ as $r$ varies. That is, we would like to find the point on an orbit that optimizes the inter-orbit speed. This can be done using the method of Lagrange multipliers  for fixed $r$ \cite{StewartBook}\cite{us_nis2_a}, which gives, for the multiplier $\nu$, the conditions
\begin{align}
\vec{b} + 2rA\hat{n} &= 2\nu \hat{n},
\end{align}
and
\begin{align}
|\hat{n}| &= 1.
\end{align}
This leads in the general case to the degree-six polynomial in $\nu$:
\begin{align}
\sum_{j=1}^3 b_j^2& (\nu-ra_{[j+1]})^2 (\nu-ra_{ [j+2]})^2 - \prod_{j=1}^3 (\nu-ra_j)^2 = 0,\label{six} 
\end{align}
where the square brackets indicate modular addition, so that the indices cycle through $1$, $2$ and $3$.

This approach will not scale up nicely to higher dimensions, since it will involve solving systems of high-degree polynomials. Instead we try a different tack with better scalability. It is easy to analyze $f_0(\hat{n}) = \vec{b}\cdot\hat{n}$. It is clear that $f_0$ is maximized at $\hat{n}_+ := \hat{b}$, minimized at $\hat{n}_- := -\hat{b}$ and zero for any vector perpendicular to $\vec{b}$. Now if we continuously increase $r$ from zero, we investigate whether there are differentiable functions $\hat{n}_\pm(r)$, with feedbacks $\vec{m}_\pm(r) := \frac{d}{dr}\hat{n}_\pm(r) $
such that $\hat{n}_\pm(r)$ are local optima for the functions $f_r(\hat{n})$ for every $r$. If such functions exist, we call the corresponding differentiable curves $\vec{n}_\pm(r) := r\hat{n}_\pm(r)$ maximizing and minimizing \emph{threads}.

We know that $\hat{n}$ lives on the sphere $S^2$, and a tangent space to $S^2$ can be identified with the two-dimensional vector space of vectors perpendicular to $\hat{n}$. The derivative of $f_r$ at a point $\hat{n}$ with respect to a variation $\vec{\epsilon}$ can be written as
\begin{align}
df_r(\hat{n})\cdot\vec{\epsilon} &= \vec{b}^T\vec{\epsilon} + 2r\hat{n}^TA\vec{\epsilon} \\
   &= (\vec{b} + 2rA\hat{n})^T \vec{\epsilon}. 
\end{align}
Note that $df_0(\pm\hat{b})\cdot\vec{\epsilon} = \vec{b}^T\vec{\epsilon} = 0 $, since $\hat{n}$ and $\vec{\epsilon}$ must be orthogonal. It follows that $\hat{n}=\pm\hat{b}$ are critical points of $f_0$, and in fact they are the only two critical points.

Our objective is to vary $\hat{n}$ with $r$ so that it remains a critical point of $f_r$. To this end, we differentiate the equation $df_r(\hat{n})\cdot \vec{\epsilon}=0$ with respect to $r$. We must be careful however since the tangent spaces at each $\hat{n}$ are different. We must also vary $\vec{\epsilon}$ so that it remains perpendicular to $\hat{n}$. Let $\vec{m}:=\frac{d\hat{n}}{dr} $ and $\vec{\mu}:= \frac{d\vec{\epsilon}}{dr}$. Since $\hat{n}(r)\cdot\vec{\epsilon}(r)=0$, $\vec{\mu}$ must always satisfy $\hat{n}\cdot\vec{\mu} = -\vec{m}\cdot\vec{\epsilon}$ (from the product rule). We now have:
\begin{align}
\frac{d}{dr}\left( df_r(\hat{n})\cdot \vec{\epsilon}\right) &= \frac{\partial}{\partial r}\left( df_r(\hat{n})\cdot \vec{\epsilon}\right) +  (2rA\vec{m})^T \vec{\epsilon} \nonumber \\
 &\hspace{1cm} + (\vec{b}+2rA\hat{n})^T \vec{\mu} \\
&= (2A\hat{n} + 2rA\vec{m})^T\vec{\epsilon} + (\vec{b}+2rA\hat{n})^T \vec{\mu}. 
\end{align}
If $\hat{n}$ is a critical point of $f_r$, the vector $\vec{b}+2rA\hat{n}$ is parallel to $\hat{n}$. Let the norm of this vector be $C$. We now have:
\begin{align}
\frac{d}{dr}\left( df_r(\hat{n})\cdot \vec{\epsilon}\right) &=(2A\hat{n} + 2rA\vec{m})^T\vec{\epsilon} + C \hat{n}^T \vec{\mu}\\
&=(2A\hat{n} + 2rA\vec{m})^T\vec{\epsilon} - C \vec{m}^T \vec{\epsilon}\\
&=(2A\hat{n} + \Lambda\vec{m})^T\vec{\epsilon},
\end{align}
where $\Lambda := 2rA-C$. If we want this expression to vanish for arbitrary $\vec{\epsilon}$ perpendicular to $\hat{n}$, we need to have, for some real $k$:
\begin{align}
2A\hat{n} + \Lambda\vec{m} &= k\hat{n} \label{rawfb}\\
\vec{m} &= \Lambda^{-1}(k-2A)\hat{n}, \label{feedback}
\end{align}
where $k$ can be found by projecting both sides on to $\hat{n}$: 
\begin{align}
k &= 2\frac{\hat{n}^T\Lambda^{-1}A\hat{n}}{\hat{n}^T\Lambda^{-1}\hat{n}}.
\end{align}
We can now state the following proposition:
\begin{proposition}
Consider a point $(r_0, \hat{n}_0) \in [0,1]\times S^2$ that is critical, in the sense that $df_{r_0}(\hat{n}_0)\cdot\vec{\epsilon}$ vanishes for any $\vec{\epsilon}$ perpendicular to $\hat{n}_0$. Define:
\begin{align}
k(r,\hat{n}) &:=  2\frac{\hat{n}^T \Lambda(r, \hat{n})^{-1}A\hat{n}}{\hat{n}^T \Lambda(r, \hat{n})^{-1}\hat{n}} \\
\Lambda(r, \hat{n}) &:= 2rA - C(r, \hat{n}) \\
C (r, \hat{n}) &:=  \hat{n}^T ( \vec{b} + 2rA\hat{n}).
\end{align}
If there is an open set $\Omega\subseteq [0,1]\times S^2$ containing $(r_0,\hat{n}_0)$, over which (1) $\Lambda(r,\hat{n})$ is invertible, and (2) $\hat{n}^T\Lambda(r,\hat{n})^{-1}\hat{n} \ne 0$, then the ODE
\begin{align}
\frac{d\hat{n}}{dr} = \Lambda(r,\hat{n})^{-1}(k(r,\hat{n}) - 2A)\hat{n}, 
\end{align}
has a unique solution $\hat{n}(r)$ on some interval $[r_0, r_f)$, and every point on this solution satisfies 
$df_r(\hat{n}(r))\cdot\vec{\epsilon} = 0, $ $\forall\epsilon$ perpendicular to $\hat{n}(r)$. 

If the conditions on $\Lambda(r,\hat{n})$ hold for all points on $[0,1]\times S^2$, there exist two threads, $\hat{n}_\pm (r)$, that satisfy the ODE and have initial conditions $(r_0,\hat{n}_{0\pm}) = (0, \pm\hat{b})$. 
\end{proposition}

\begin{proof}
The local existence and uniqueness result is an application of  standard theory of ODE's on manifolds \cite{LangBook}. Local existence requires differentiability of the RHS, which clearly holds given the two conditions. The vanishing of the derivative has alreadly been shown. The global existence result holds because $[0,1] \times S^2$ is compact.
\end{proof}

Both of the conditions on $\Lambda(r,\hat{n})$ sometimes fail, which we shall discuss in section V.  The first is always tractable, while analysis of the second requires the solution of a polynomial. Apart from these special cases however, we now have two functions $\hat{n}_\pm(r)$ with $\hat{n}_\pm (0)=\pm\hat{b}$ and $\frac{d}{dr}\hat{n}_\pm = \vec{m}(\hat{n}_\pm)$ that optimize (at least locally) $f_r(\hat{n})$ for each $r$. 

\section{Trajectory planning}

We now have a feedback $\vec{m} = \frac{d\hat{n}}{dr}$ that ensures a trajectory remains on a critical point as it moves inward to or outward from the completely mixed orbit. Our control variable however is the Hamiltonian, so to plan a trajectory (for example, to move from the completely mixed state $r=0$ outwards), we must know how to recover an appropriate Hamiltonian from a given feedback. 

We know previously that $\frac{d\hat{n}}{dt} = -\frac{1}{r}\vec{b}_\perp + \vec{h}\times\hat{n} - (A\hat{n})_\perp$. We also know $\frac{d\hat{n}}{dr} = \vec{m}$. If we write $\frac{d\hat{n}}{dt} f(r, \hat{n}) = \frac{d\hat{n}}{dt}$, we get 
\begin{align}
\vec{m}(r,\hat{n})f(r,\hat{n})& = \vec{h} \times\hat{n} - \frac{1}{r}\vec{b}_\perp - (A\hat{n})_\perp\\
\vec{h} &= \hat{n} \times\left(\vec{m}(r,\hat{n})f(r,\hat{n}) + \frac{1}{r}\vec{b}_\perp + (A\hat{n})_\perp \right).
\end{align}
Note that the term $\frac{1}{r}\vec{b}_\perp$ is well-behaved since $\vec{b}\times\hat{n}$ is zero at $r=0$ and $O(r)$ in the vicinity. Also note that an arbitrary component parallel to $\hat{n}$ can be added to $\vec{h}$. In short, to attain a desired trajectory $\hat{n}(r)$ using its associated feedback $\vec{m}(r,\hat{n})$, one should apply a radially varying Hamiltonian in the form
\begin{align}
\vec{h}(r) = c(r)\hat{n}(r) + \hat{n}(r) \times \left(\vec{m}(r,\hat{n}(r)) + \frac{1}{r}\vec{b} + A\hat{n}(r) \right),
\end{align}
where $c(r)$ is arbitrary.

One might think the piece including $\vec{m}$ may blow up if $\Lambda$ becomes non-invertible, but we will see that this is not the case. There are cases where $\vec{m}$ does blow up, but this does not occur on the main threads $\hat{n}_\pm(r)$ that arise at $r=0$. Instead, this occurs when alternate threads arise at some $r > 0$. We discuss these possibilities in the next section. 

\section{Special cases}

We now consider the instances in which the feedback may not be well-defined, either due to (1) $\Lambda$ losing invertibility, or (2) the denominator in the definition of $k$ vanishing. The matrix $\Lambda$ loses invertibility if and only if the constant $C$ equals one of the eigenvalues $a_j$ of $2rA$. But at any critical point, we have:
\begin{align}
\vec{b} + 2rA\hat{n} &= C\hat{n} \\
\vec{b} &= -\Lambda\hat{n}.
\end{align}
In other words, degeneracy of $\Lambda$ implies that $\vec{b}$ must be in the image of $\Lambda$, which does not have full dimension. We separate the cases based on the multiplicity of the eigenvalue of $2rA$ in question. \\

\textbf{(1a) $C$ is a triple eigenvalue of $2rA$.}\\ 
This case is largely trivial. If $A$ is a multiple of the identity $aI_2$, where $a$ is the only eigenvalue of $A$, and $C = 2ra$, then $\Lambda = 0$, and therefore $\vec{b}=0$. We have then:
\begin{align}
f_r(\hat{n}) &= -2ar \\
df_r (\hat{n})\cdot\vec{\epsilon} &= 2ra \hat{n}^T\vec{\epsilon} = 0.
\end{align}
Therefore all possible $\hat{n}$ are critical points of $f_r$ due to the rotational symmetry. No optimization is needed since $f_r$ is a constant function. \\

\textbf{(1b) $C$ is a double eigenvalue of $2rA$.}\\
In this case, we find a plane of critical points that intersect one of the main threads $\hat{n}_\pm(r)$. Additionally, this case covers the well-studied phase- and amplitude-damping channels.

Let $\{e_j\}$, $j = 1,2,3$, be an eigenbasis of $A$, where the eigenvalues $a_1$ and $a_2$ are equal, and $a_3\ne a_1$. It follows that for $C=2ra_1$, we need $\hat{b}=e_3$. We know critical points satisfy $\vec{b}+2rA\hat{n} = C\hat{n}$, which gives:
\begin{align}
2ra_1\hat{n}_j &= C\hat{n}_j, \hspace{5pt} j = 1,2\\
b_3 + 2ra_3\hat{n}_3 &= C\hat{n}_3.
\end{align} 
There are two solutions. If $\hat{n}_1 = \hat{n}_2 = 0$, then $\hat{n}_3 = \pm 1$. This solution corresponds to the main threads $\hat{n}_\pm(r) = \pm \textrm{sgn}(b_3)$. In this case, no feedback is needed, since eq. (\ref{rawfb}) is satisfied for $\vec{m}=0$. As it happens, $C = \pm b_3 +2ra_3$, which means that $\Lambda_\pm = \mp b_3I + 2r(a_1-a_3) \cdot\textrm{diag}(1, 1, 0)$. Invertibility is lost at $r = \frac{|b_3|}{2|a_1-a_3|}$, but yet the feedback solution $\vec{m}=0$ is still valid there.

An alternate solution exists. If $C=2a_1r$, then $\hat{n}_1$ and $\hat{n}_2$ are free. In this case, $n_3 = r\hat{n}_3 = \frac{b_3}{2(a_1-a_3)}$. So there is a plane of critical points, that happens to orthogonally intersect one of the main threads at exactly the point where the corresponding $\Lambda$ loses invertibility. On this critical plane, $\Lambda$ is everywhere non-invertible: it is equal to $\textrm{diag}(0, 0, 2r(a_3-a_1))$. Yet any $\vec{m}$ with $m_3 = \frac{b_3}{2r^2(a_3-a_1)}$ satisfies eq. (\ref{rawfb}). This solution allows a valid $\vec{m}$, unless $\hat{n}$ approaches the intersection point with the main thread. Since $\vec{m}$ and $\hat{n}$ must be perpendicular, some combination of the components $m_1$ and $m_2$ must grow unbounded, and at the intersection point itself, there is no solution since perpendicularity forces $m_3=0$. This also implies that one cannot switch from the main thread to the alternate plane. 

We can apply this to a combination of phase- and amplitude-damping channels \cite{Preskill}. A phase-damping channel uses a Lindblad operator in the form $L_z := \sqrt{\gamma_z}\sigma_z$, while the amplitude-damping channel uses Lindblad operators in the form $L_\pm := \sqrt{\gamma_\pm}\sigma_\pm$, where $\sigma_\pm = \frac{\sigma_x \mp i\sigma_y}{2}$. In this case, we get the following parameters: $a_1=\gamma_z$, $a_2 = a_3 = \frac{\gamma_++\gamma_-}{4}$ and $\vec{b}= \langle \frac{\gamma_+ -\gamma_-}{2}, 0, 0\rangle$. Since $b_2=b_3=0$ and $a_2=a_3$, we have a plane of critical points at $e_1\cdot \vec{n} = \frac{\gamma_+-\gamma_-}{4\gamma_z - \gamma_+-\gamma_-}$. \\

\textbf{(1c) $C$ is a single eigenvalue of $2rA$.}\\
This case is similar to the preceding, except the plane of critical points is now a line of critical points. Let $C = 2ra_1$ and $a_1\ne a_2, a_3$, we still have $b_1=0$. We get:
\begin{align}
2a_1r\hat{n}_1 &= C\hat{n}_1 \\
b_j + 2a_jr\hat{n}_j &= C\hat{n}_j, \hspace{5pt} j= 2,3
\end{align}
As before, choosing $\hat{n}_1 = 0$ allows us to recover the main threads $\hat{n}_\pm(r)$.  $C$ will equal the offending value $2ra_1$ at $r=\sqrt{\frac{b_2^2}{4(a_1-a_2)^2}+\frac{b_3^2}{4(a_1-a_3)^2}}$. The fact that $\Lambda$ loses invertibility here does not affect the feedback, because the direction of degeneracy happens to be orthogonal to $\hat{n}$: that is, $\Lambda$ is degenerate in the $e_1$-direction. Therefore the feedback can still be found.

An alternate thread can be found by setting $C = 2ra_1$ and letting $\hat{n}_1$ run free. In this case, we find that $n_j = r\hat{n}_j = \frac{b_j}{2(a_1-a_j)}$, so this thread is orthogonal to the $\hat{n}_1 = 0$ plane. While $\Lambda$ is non-invertible on this thread, we can find a solution to eq. (\ref{rawfb}): $m_1$ can be free, while $m_2=-n_2$ and $m_3=-n_3$. To satisfy perpendicularity, however, we require $m_1 = \frac{n_2^2+n_3^2}{n_1}$. This clearly blows up as the thread crosses the $n_1=0$ plane.

Note the alternate thread may or may not intersect the main threads, but from our simulations, we observe that intersections only seem to occur when $a_2=a_3$.  \\

\textbf{(2) $k$ is not well-defined.}\\
The above cases are simpler than a generic system, since a component of $\vec{b}$ vanishes in the natural co-ordinates of $A$, which reduces the degree of equation (\ref{six}) from six to four (or two). However, it may still happen that $\Lambda$ is invertible, yet $\hat{n}^T\Lambda^{-1}\hat{n} = 0$. In this case, the algebra required to find such a location still leads to a degree-six polynomial. We essentially have five unknowns: $r$, $C$ and the three components of $\hat{n}$. These obey five equations: 
\begin{align}
b_j+2ra_j\hat{n}_j &= C\hat{n}_j \\
|\hat{n}| &= 1 \\
\sum_j \frac{\hat{n}_j^2}{2ra_j-C}  &= 0, 
\end{align}
where the final equation is the failure of condition (2). We can eliminate the variable $r$ and the second equation by working with the components of $\vec{n}$ instead of $\hat{n}$. Furthermore, working with $\mu = \frac{C}{2r}$ yields $n_j = \frac{b_j}{2(a_j-\mu)}$. Substitution into the third equation gives $\sum_j \frac{b_j^2}{8(a_j-\mu)^3} = 0$. This yield the sixth-degree polynomial equation:
\begin{align}
b_1^2 (a_2-\mu)^3(a_3-\mu)^3 + b_2^2 (a_3-\mu)^3(a_1-\mu)^3  \nonumber \\
 + b_3^2 (a_1-\mu)^3(a_2-\mu)^3  = 0.
\end{align}
One can find solutions numerically for $\mu$ and the corresponding $\vec{n}$ follows easily. Such a solution corresponds to an alternate thread of critical points: when such a thread becomes tangent to a concentric sphere in the Bloch ball, the feedback $\vec{m}$ becomes infinite, which is why the feedback expression fails. It is possible to plot such an alternate thread by locating an initial point away from where the feedback fails, and then using the feedback in either direction. To locate such a point, one needs to solve the degree-six polynomial (\ref{six}). We will not do this in our examples, as it contradicts the spirit of this paper. In higher dimensions, the algebra would not be tractable, therefore we must make peace with the fact that the feedback works only to find critical points locally. 

\section{Separation of the Bloch Ball}

Besides finding the optimal points of $f_r$, it is also an interesting question to locate the zeros of $f_r$. It turns out that for $\vec{b}\ne 0$, there is a ``chimney" region in the Bloch ball, inside of which the purity ``rises". That is, $f(\hat{n},r) > 0$, and outside of which $f(\hat{n} , r) < 0 $. If we want $r$ to increase, we must steer inside of this region. We can locate the ``wall" of this chimney by using another  feedback expression. We know that at $r=0$, we have $f_0(\hat{n})=0$ for $\hat{n}\cdot\vec{b}=0$, which has a $S^1$-homeomorphic set of solutions, say $\hat{c}_\theta$, with $\theta\in [0,2\pi)$ being an angle parameter. We want to take such a solution, use it as an inital condition, and find a feedback to ensure $f_r(\hat{n})=0$ as $r$ increases. To do this, we differentiate $f_r$ with respect to $r$, with $\vec{m}=\frac{d\hat{n}}{dr}$:
\begin{align}
\frac{d}{dr} \left(f_r(\hat{n}) \right) &= \frac{\partial}{\partial r}f_r(\hat{n}) +\nabla_{\hat{n}}f_r(\hat{n})\cdot \vec{m} \\
&= \hat{n}^TA\hat{n} - \textrm{tr}(A) + (\vec{b} + 2rA\hat{n})^T \vec{m} \\
(\vec{b} + 2rA\hat{n})^T \vec{m} &=  \textrm{tr}(A)  - \hat{n}^TA\hat{n} .
\end{align}
To satisfy this equation, as well as $\vec{m}\cdot\hat{n}=0$, define $\vec{v} = \vec{b}+2rA\hat{n}$. A possible solution is:
\begin{align}
\vec{m} = \frac{\textrm{tr}(A)  - \hat{n}^TA\hat{n}}{|\vec{v}|^2 - (\hat{n}\cdot\vec{v})^2}\left(\vec{v} - (\hat{n}\cdot\vec{v})\hat{n }\right). \label{chimney}
\end{align}
This solution is not unique: for a given $r$, there is a continuum of zeros of $f_r$, at least until the chimney terminates. If we kept $r$ fixed, and moved along this continuum with $\hat{n} = \hat{n}(t)$, a feedback $\frac{d\hat{n}}{dt} \propto \hat{n}\times\vec{v}$ would ensure $f_r(\hat{n}(t)) = 0$. For our feedback, we will thus keep the component parallel to $\hat{n}\times\vec{v}$ zero, so that we capture only the necessary motion of $\hat{n}$.

Of course, this feedback will terminate for some $r\le 1$, since $f_r$ cannot be positive at that radius. The terminating condition is $\hat{n}\cdot\vec{v} = |\vec{v}|$, which matches the critical point condition. At such a point, the feedback becomes infinite. Thus the point on the chimney furthest from the origin (which we call the \emph{apogee}) is a critical point, either on the maximizing thread $\hat{n}_\pm$, or possibly on one of the alternate threads. In the following section, we will show examples of both possibilities. 

Finally, it should be noted that the chimney does have an analytic solution. If one uses $rf_r(\hat{n}) = 0$ and substitutes $r^2 = \vec{n}^2$, we obtain an elllipsoid in the co-ordinates of $\vec{n}$:
\begin{align}
\sum_j \tilde{a}_j \left(n_j - \frac{b_j}{2\tilde{a}_j}\right)^2 = \sum_j \frac{b_j^2}{4\tilde{a}_j},
\end{align}
where $\tilde{a}_1 := a_2 + a_3$ and so forth. In general however, the ellipsoid center and axes are not aligned with the axes of $A$, other than intersecting the center of the Bloch ball. So the intersection of the ellipsoid with concentric spheres, which is what we are interested in, will not have a clean analytic expression. In fact the intersection may not even be connected: this is what happens when there is more than one apogee.

\section{Examples}

\begin{figure}
\includegraphics[width=0.9\columnwidth]{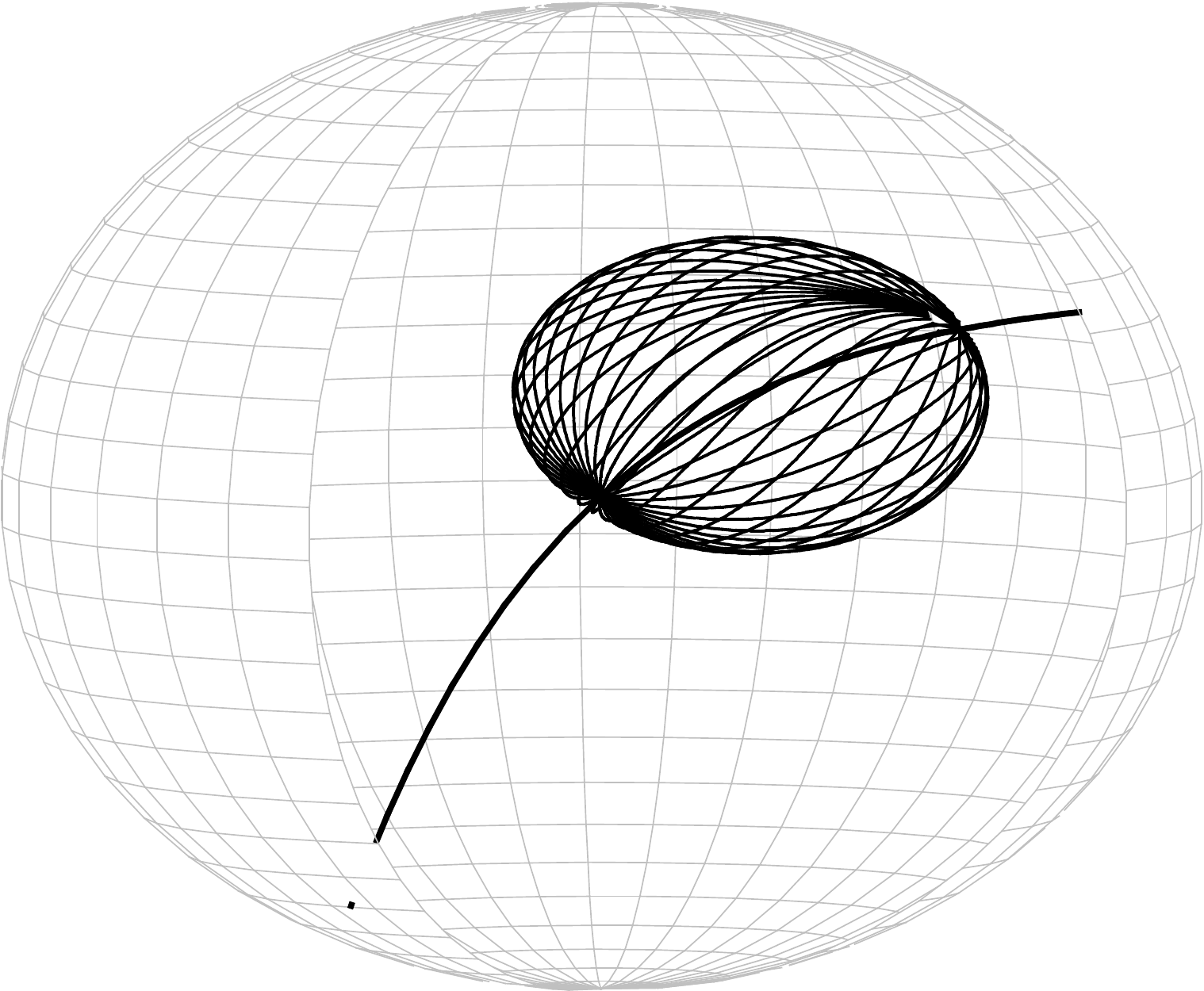}
\caption{Optimizing threads and chimney for $A=\textrm{diag}(100, 57, 39)$ and $b=\langle 29,67, 61\rangle$. Inside the chimney, the dynamics drives the Bloch vector outwards, whereas Bloch vectors outside have negative radial velocity. The thread represents points that have optimal radial velocity. The piece from the center to the lower left has minimal velocity, while the piece in the upper right has maximal velocity.} 
\label{fig1}
\end{figure}

The feedbacks (\ref{feedback}) and (\ref{chimney}) can be used to form ODE's $\frac{d}{dr}\hat{n} = \vec{m}(r, \hat{n})$ with initial conditions $\hat{n}_\pm (0) = \pm\hat{b}$ or $\hat{n}_\theta(0) = \hat{c}_\theta$ that can be solved numerically. The chimney can be plotted by discretizing the circle of initial points and calculating threads on the chimney. We have implemented this using a Runge-Kutta method for Lie groups \cite{Celledonietal2012} which ensures $\hat{n}$ remains normalized. The results are consistent with the preceding analysis. Fig. \ref{fig1} shows a typical example. The interval $r\in[0,1]$ has been discretized into intervals of length $\frac{1}{1000}$. The maximizing thread curls towards the upper right, and the minimizing thread to the lower left. We have estimated the error by calculating the component of $\vec{b}+2rA\hat{n}$ perpendicular to $\hat{n}$, and we can report that this error does not exceed $3\times 10^{-10}$ for either thread in this example. Typically a discretization of $\delta r = \frac{1}{1000}$ is sufficient to achieve precision of such order. 

The chimney is also plotted by discretizing the circle of initial points into thirty-six. It is important to note that the algorithm is not capable of finding the apogee of the chimney, since the ODE blows up there. One must stop the algorithm when the error exceeds a certain threshold. For the chimney we estimate the error by calculating $f_r(\hat{n}_\theta(r),r)$, and we can report for this example the error does not exceed $2.5\times 10^{-6}$. The threshold we used was $1\times 10^{-3}$. For this example, the chimney threads finish near the maximizing thread, so we can infer that their termination point lies on this thread. The termination point can be calculated by finding the zero of $f(\hat{n}_+(r),r)$.  There are no alternate threads for this example.

In fig. \ref{fig2}, we have an example with an alternate thread. Since $b_1=0$ and $a_2=a_3$, we know there will be a line of critical points that intersects the maximizing thread. The alternate thread of critical points is horizontal in the plot, with thinner line-width. When we plot the chimney we can see that all but two of the thirty-six chimney threads terminate on the alternate thread, rather than the maximizing thread. Below the Bloch ball we also plot $f(\hat{n}_\pm(r),r)$ with thick line-width and $f(\hat{n}_a(r),r)$ with thin line-width, where $\hat{n}_a$ is the alternate thread. We can see that the maximizing thread only gives a local maximum for radii at which the alternate thread exists, and the alternate thread provides the global maximum.

\begin{figure}
\includegraphics[width=0.9\columnwidth,trim=0cm -2cm 0cm 0cm]{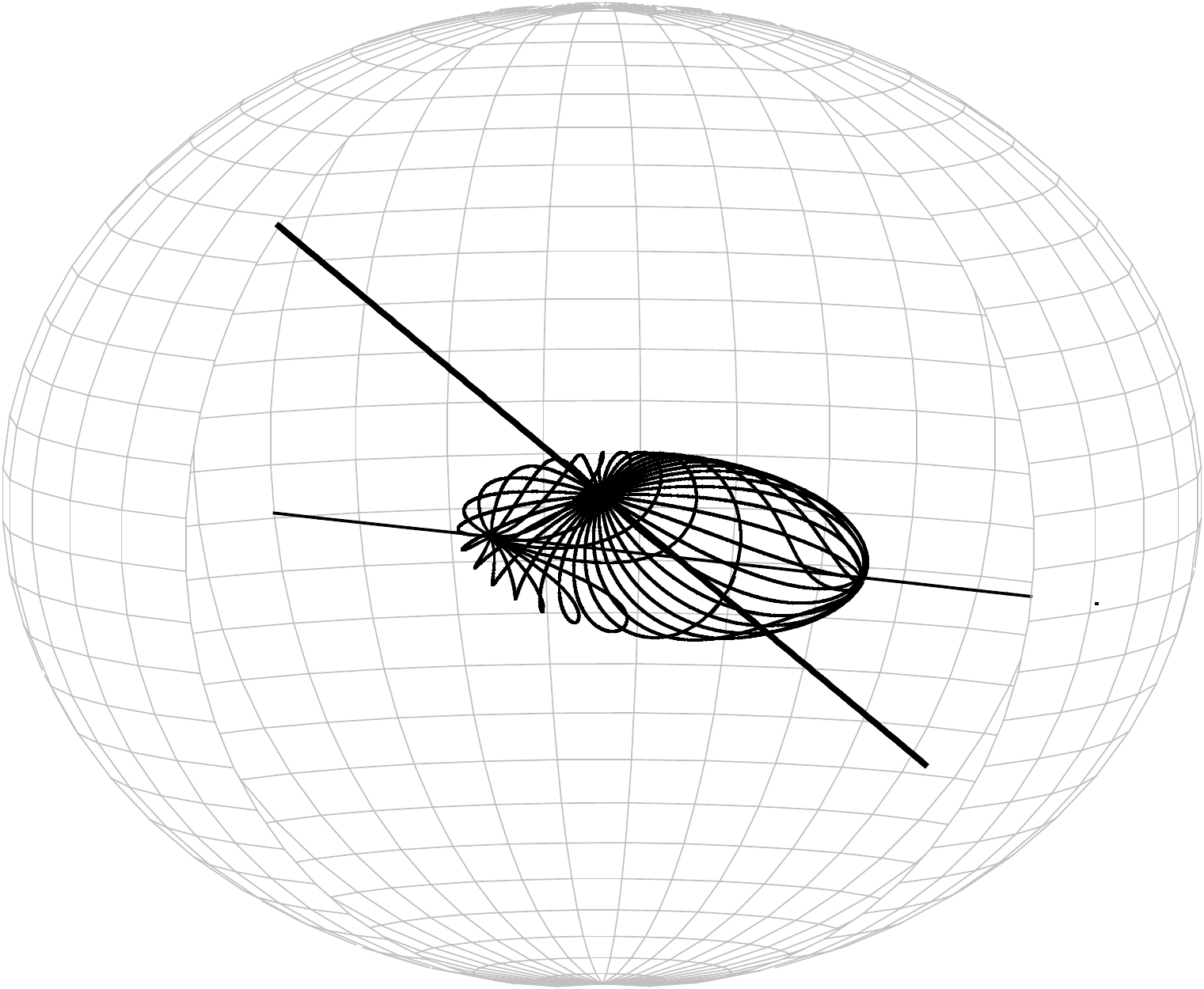}
\includegraphics[width=0.9\columnwidth]{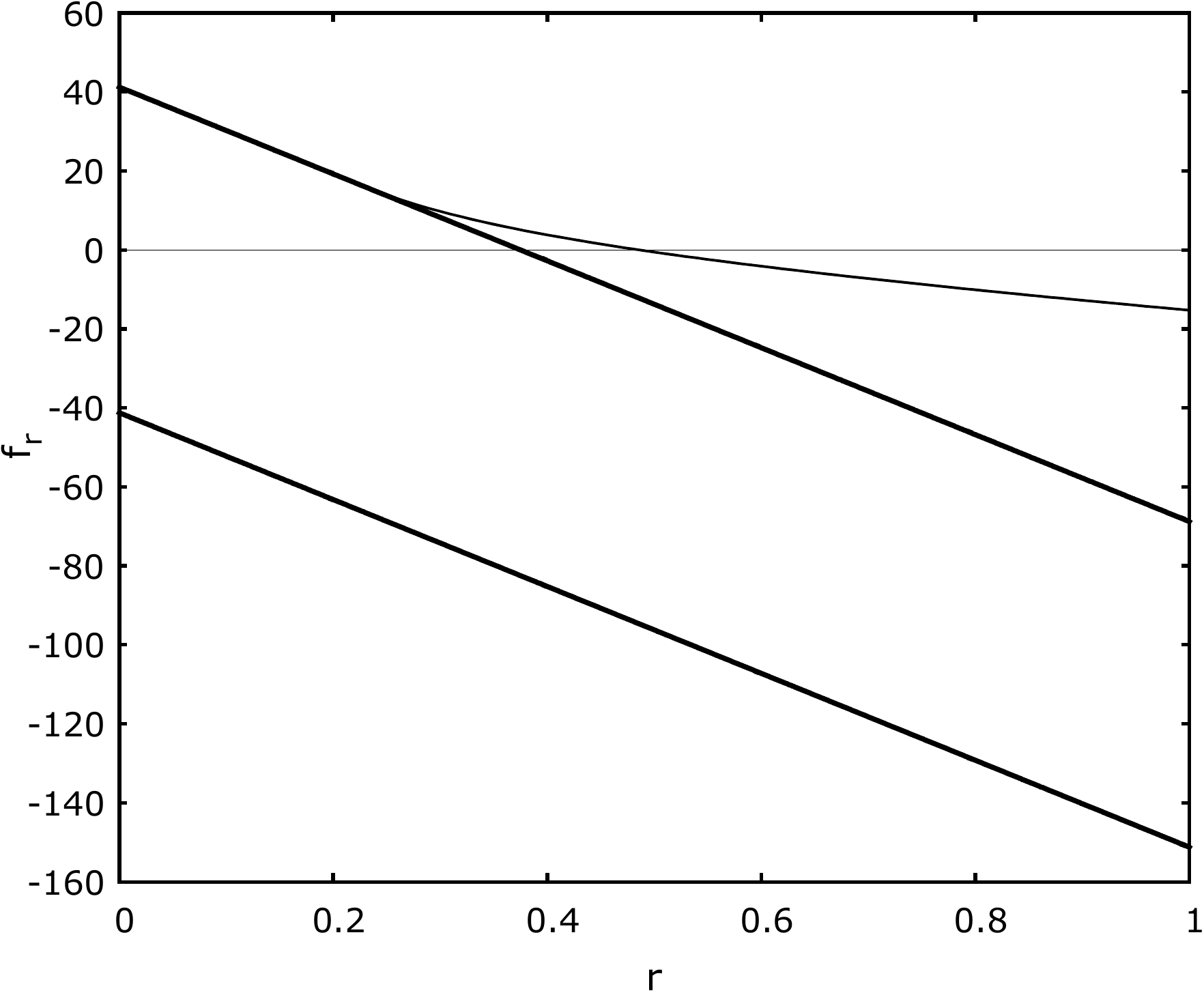}
\caption{(Top) The main threads are represented by the thick line, while the transverse line represents an alternate thread. $A=\textrm{diag}(100, 10, 10)$ and $b=\langle 0,32, -26 \rangle$. (Bottom) $f(\hat{n}_\pm, r)$ and $f(\hat{n}_a, r)$ for the same system. Thick lines represent radial velocity of the main threads, while the thinner line represents the radial velocity of the alternate thread. Clearly, the alternate thread is the global maximum on its domain. } 
\label{fig2}
\end{figure}

In fig. \ref{fig3}, we have another example with an alternate thread. This time $b_2=0$ and $a_1\ne a_3$, and and we see the line of critical points does not intersect the optimizing threads. When we plot $f(\hat{n}_\pm(r),r)$ and $f(\hat{n}_a(r),r)$, we can see the alternate thread does not provide a global optimum, and the optimizing threads provide global optima for all $r$.

\begin{figure}
\includegraphics[width=0.9\columnwidth,trim=0cm -2cm 0cm 0cm]{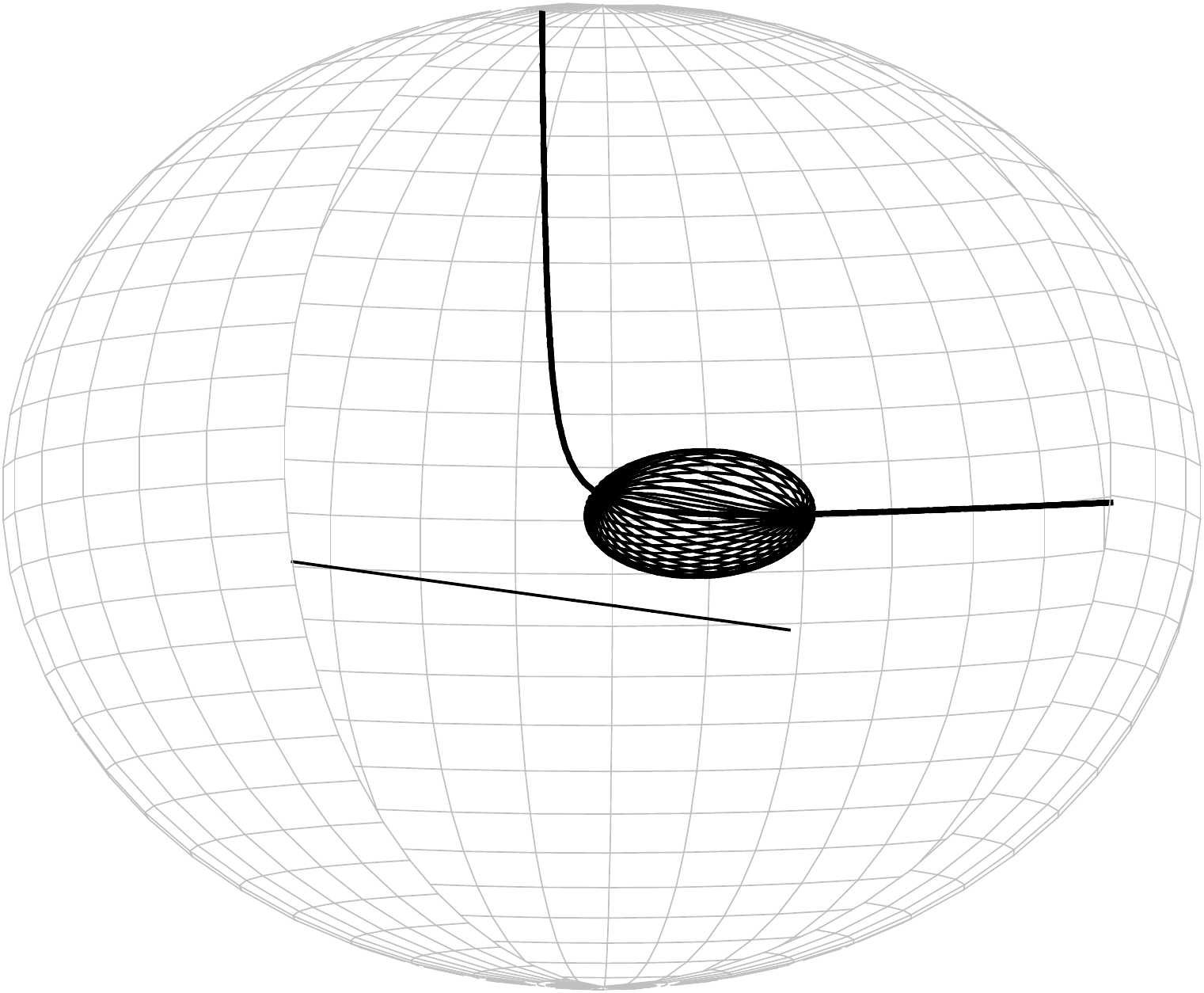}
\includegraphics[width=0.9\columnwidth]{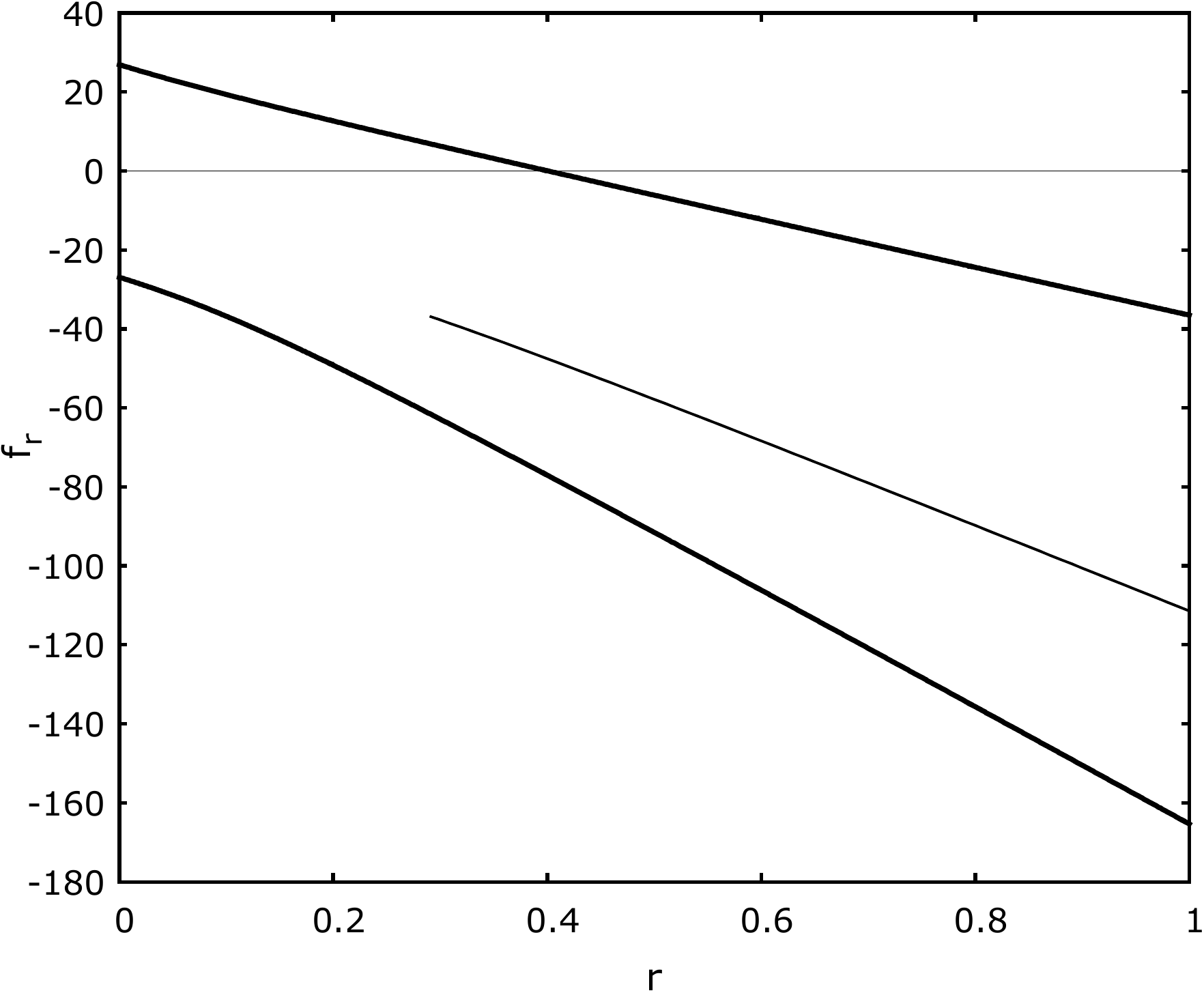}
\caption{(Top) The main threads that pass through the chimney are shown in bold, while an alternate thread also exists. There is no intersection, as $a_2\ne a_3$. $A=\textrm{diag}(100, 50, 10)$ and $b=\langle 23,0, -14 \rangle$. (Bottom) $f(\hat{n}_\pm, r)$ and $f(\hat{n}_a, r)$ for the same system. The thinner line represents the radial velocity of the alternate thread. Clearly the alternate thread is not a global optimum.} 
\label{fig3}
\end{figure}

In fig. \ref{fig4}, we have an example where $b_j\ne 0$, and yet there is still an (unshown) alternate thread. While nineteen of the thirty-six chimney threads terminate on the maximizing thread, the remaining seventeen clearly terminate elsewhere, and so the chimney has a second apogee. In fact, there is an alternate thread that begins inside the chimney and exits at this hole (there is another exit point that does not serve as a termination point, because it is a saddle point). In keeping with the spirit of this paper, we have not attempted to plot this alternate thread or determine whether it provides a global optimum. We can report that the termination point on the maximizing thread is at a larger radius ($r\cong 0.748$) than the alternate termination point ($r\cong 0.649$).

We have however decided to estimate how often a system has an alternate thread. We have simulated 100,000 random systems in the following way: the largest eigenvalue of $A$ was fixed to be $a_1=100$. The remaining two were chosen to be uniform on the interval $[0,100]$. To randomize $\vec{b}$ we know that, due to the positive-definiteness of the GKS matrix \cite{GoriniKossakowskiSudarshan76}, it obeys the inequality (\ref{ineq}).
Thus the vector $\vec{b}^* = \langle \frac{b_1}{2\sqrt{a_2a_3}}, \frac{b_2}{2\sqrt{a_1a_3}}, \frac{b_3}{2\sqrt{a_1a_2}}\rangle$ must lie in a ball of radius one. We impose a uniform distribution on this ball, choose a $\vec{b}^*$ and calculate $\vec{b}$. With this randomization, we conducted 100,000 simulations that yielded 59,830 systems without an alternate thread, 30,811 with one alternate thread and 9,359 with two. 

\begin{figure}
\includegraphics[width=0.9\columnwidth]{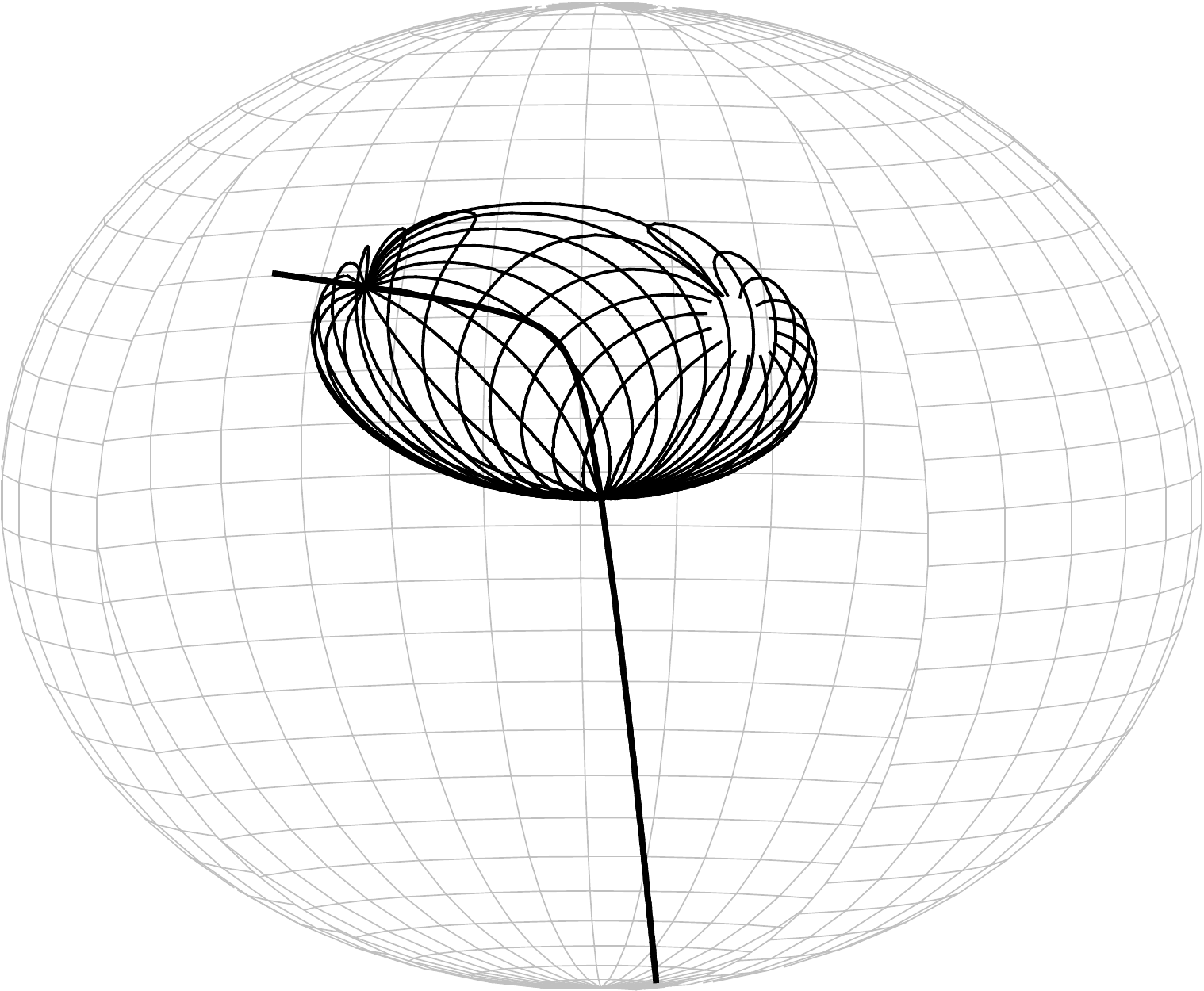}
\caption{A system with main threads, and and unshown alternate thread. The main threads optimize radial velocity locally. However, the fact that chimney lines approach a point that is clearly not on the main threads indicates there is an alternate thread that also locally maximizes radial velocity.  $A=\textrm{diag}(100, 16, 11)$ and $b=\langle -3, -8, 68 \rangle$. } 
\label{fig4}
\end{figure}

\section{Conclusions and Future Work}

We have demonstrated that it is possible to derive a feedback equation that maintains critical points and zeros as one transitions between quantum orbits. The behavior at the completely mixed state is easy to analyze: both the critical points and the zeros of the function $f(r,\hat{n}) = \frac{dr}{dt}$ are trivial to compute at that orbit. As one increases $r$, these zeros and critical points can be preserved. The critical points form two threads: one of which maximizes $\frac{dr}{dt}$ locally, the other minimizes. If one has fast controllability and one wants to optimize the speed at which the state moves between orbits, the system can be steered to either of these threads, depending on the desired direction. The feedback expression also yields an expression for a Hamiltonian that keeps the system on the thread. 

It is important to note that this mechanism only ensures that the optima are local. There are systems where other optima emerge as one moves away from the completely mixed state. Sometimes such an alternate optimum is also the global optimum, sometimes not. 

The intention of this paper is to demonstrate an approach that can be used in higher dimensions to analyze controllability. Because the Lindblad term $\mathcal{L}_D(\rho)$ in dimension $n$ reduces to $\frac{1}{n}\sum_m [L_m,L_m^\dagger]$ at the completely mixed state, it can be treated analytically, using the Schur-Horn theorem \cite{Schur}\cite{Horn}. Hopefully, one can study the critical points and zeros of $\mathcal{L}_D$ away from the complete mixed state by using a feedback similiar to the method used in this paper. Instead of $\hat{n}$, one considers the flag formed by the eigenstates of $\rho$. Such a flag can be made to vary continuously by applying a skew-Hermitian operator: its tangent space is a subspace of the Lie algebra $\mathfrak{su}(n)$. It is reasonable to assume that one can achieve a feedback expression on the tangent space that preserves critical points and zeros. A paper examining these ideas is in preparation.


\appendix
 \section{Derivation of the Bloch vector ODE}

A quantum density operator $\rho$ is a trace-one, positive semi-definite operator. On an $n=2$ Hilbert space, we can write:
\begin{align}
\rho = \frac{1}{2}\left(I_2 + \sum_{j=x,y,z}n_j \sigma_j\right),
\end{align} 
where $\sigma_j$ are the Pauli matrices:
\begin{align}
\sigma_x = \left(\begin{array}{cc}
0&1\\
1&0
\end{array}\right), \hspace{0.1in}
\sigma_y = \left(\begin{array}{cc}
0&-i\\
i&0
\end{array}\right),\hspace{0.1in}
\sigma_z = \left(\begin{array}{cc}
1&0\\
0&-1
\end{array}\right),\end{align}
which obey the following relations:
\begin{align}[\sigma_j, \sigma_k] = 2i\epsilon_{jkl}\sigma_l \\
\{\sigma_j, \sigma_k\} =2\delta_{jk}I_2.
\end{align}
It can be checked that the purity $Tr(\rho^2)$ is equal to the magnitude of the Bloch vector $r:=|\vec{n}|$. It can also be shown that the eigenvalues of $\rho$ are $\lambda_\pm := \frac{1\pm r}{2}$. Each unitary orbit $\{ U \rho U^\dagger: U\in U(2)\}$ corresponds to one value of $r\in [0,1]$. 

If for $r\ne 0$, we write $\hat{n}:= \vec{n}/r$, the eigenvectors of $\rho$ are $|\psi_\pm\rangle:= \frac{\hat{n}_z\pm1}{2}|0\rangle + \frac{\hat{n}_x+i\hat{n}_y}{2}|1\rangle$. It follows that the set $\{|\psi_+\rangle\}$ can be identified with the set $\{\hat{n}\}$, which of course is $S^2$.

Using the above identification, we can transform the Lindblad equation to an ODE on $\mathbb{R}^3$. If we set $H = h_0I_2 + \frac{1}{2}\sum_j h_j\sigma_j$, the Hamiltonian piece becomes:
\begin{align}
[-iH, \rho] &= \frac{1}{4} \sum_{j,k} h_jn_k[-i\sigma_j,\sigma_k] = \frac{1}{2}\sum_{j,k} h_jn_k \epsilon_{jkl}\sigma_l \\
&= \frac{1}{2}\sum_l \left(\vec{h}\times\vec{n}\right)_l\sigma_l .
\end{align}
We can assume the Lindblad operators are traceless, as any traced part can be absorbed into the Hamiltonian \cite{BreuerPetruccioneBook}. In this case, we can write $L_m = \sum_{j=x,y,z}l_{j,m}\sigma_j$, where $l_{j,m}\in\mathbb{C}$. We have: 
\begin{align}
\mathcal{L}_D(\frac{I_2}{2}) &= \frac{1}{2}\sum_{j,k,m}l_{j,m}\bar{l}_{k,m}[\sigma_j,\sigma_k] = \sum_{j,k,m}l_{j,m}\bar{l}_{k,m} i\epsilon_{jkl}\sigma_l \\
&= \sum_{l} b_l\sigma_l,
\end{align}
where
\begin{align}
\vec{b} &= i \sum_{m} \vec{l}_m \times \vec{\bar{l}}_m.
\end{align}
If all Lindblad operators are Hermitian, $\vec{b}$ vanishes.  This is known as the unital case. 

We also have: 
\begin{align}
\mathcal{L}_D(\sum_l&\frac{n_l\sigma_l}{2} ) = \sum_{j,k,l,m}l_{j,m}\bar{l}_{k,m}\frac{n_l}{4} \left(2\sigma_j\sigma_l\sigma_k \right. \nonumber \\ &\hspace{0.5in} \left.- \sigma_k\sigma_j\sigma_l-\sigma_l\sigma_k\sigma_j\right) \\
&= \frac{1}{4}\sum_{j,k,l,m}l_{j,m}\bar{l}_{k,m}n_l \left(\delta_{jl}\sigma_k+\delta_{kl}\sigma_j-2\delta_{jk}\sigma_l \right) \\
&= \frac{1}{2}\sum_{j,l,m} \frac{l_{l,m}\bar{l}_{j,m} +l_{j,m}\bar{l}_{l,m}}{2} n_l\sigma_j - l_{j,m}\bar{l}_{j,m} n_l\sigma_l \\
&= \frac{1}{2} \sum_{l} \left( A\vec{n}\right)_l \sigma_l - \textrm{tr}(A) n_l\sigma_l,
\end{align}
where $A$ is the symmetric matrix 
\begin{align}
A &:= \frac{1}{2}\sum_{m}(\vec{l}_m\vec{\bar{l}}_m^T+\vec{\bar{l}}_m\vec{l}_m^T).
\end{align} 

Since $\frac{d}{dt}\rho = \frac{1}{2}\sum_{j}\frac{dn_j}{dt}\sigma_j$, we can combine these pieces into the following ODE:
\begin{align}
\frac{d\vec{n}}{dt} = \vec{b} + \vec{h}\times\vec{n} +(A - \textrm{tr}(A))\vec{n}.
\end{align}

\section{Parameter conditions}

Since $A$ is a symmetric matrix, it has a natural orthonormal basis. In this basis, we have six system parameters: the eigenvalues $\{a_j\}$ of $A$ and the elements $\{b_j\}$ of $\vec{b}$, with $j = 1,2,3$. These six parameters must obey two inequalities. 

Consider the matrix $A_* = \sum_m \vec{l}_m\vec{\bar{l}}^T_m$. $A_*$ is the sum of positive semi-definite matrice, and so itself must be positive semi-definite. Moreover, its real part, which equals $A$, must be positive semi-definite, so we have our first inequality:
\begin{align}
a_j \ge 0
\end{align}
Now the imaginary part of $A_*$ relates to $\vec{b}$: $b_1 = i (l_2\bar{l}_3 - l_3\bar{l}_2) = 2\hspace{3pt} \textrm{Im} (A_*)_{32}$ etc. If we write $A_*$ in the natural basis of $A$, and take its determinant, we get:
\begin{align}
\det(A_*) = a_1a_2a_3-\frac{1}{4}(a_1b_1^2+a_2b_2^2+a_3b_3^2).
\end{align}
Since the determinant of a positive semi-definite matrix must be non-negative, we recover the second inequality:
\begin{align}
\vec{b}^TA\vec{b} \le 4 \det(A).
\end{align}

\begin{acknowledgments}
P.R. has been supported by the National Science Foundation and the DFG grant HE 1858/13-1 from the German Research Foundation (DFG). A.M.B. is supported by the National Science Foundation. C.R. is supported by the Natural Science and Engineering Research Council of Canada.
\end{acknowledgments}


\begin{thebibliography}{10}

\bibitem{ShapiroBrumer86}
M.~Shapiro and P.~Brumer.
\newblock Laser control of product quantum state populations in unimolecular
  reactions.
\newblock {\em J. Phys. Chem.}, 84(7):4103, 1986.

\bibitem{TannorRice85}
D.~J. Tannor and S.~A. Rice.
\newblock Control of selectivity of chemical reaction via control of wave
  packet evolution.
\newblock {\em J. Chem. Phys.}, 83(10):5013, 1985.

\bibitem{ErnstetalBook}
R.~R. Ernst, G.~Bodenhausen, and A.~Wokaun.
\newblock {\em Principles of Nuclear Magnetic Resonance in One and Two
  Dimensions}.
\newblock Clarendon, Oxford, 1987.

\bibitem{RanganBucksbaum01}
C.~Rangan and P.~H. Bucksbaum.
\newblock Optimally shaped terahertz pulses for phase retrieval in a
  {R}ydberg-atom data register.
\newblock {\em Phys. Rev. A}, 64:033417, 2001.

\bibitem{PalaoKosloff02}
J.~P. Palao and R.~Kosloff.
\newblock Quantum computing by an optimal control algorithm for unitary
  transformations.
\newblock {\em Phys. Rev. Lett.}, 89(18):188301, 2002.

\bibitem{MabuchiKhaneja2005}
H.~Mabuchi and N.~Khaneja.
\newblock Principles and applications of control in quantum systems.
\newblock {\em Int J. Robust and Nonlinear Control}, 15:647 -- 667, 2005.

\bibitem{BrifChakrabartiRabitz2010}
Brif, Chakrabarti, and H.~Rabitz.
\newblock Control of quantum phenomena: past, present and future.
\newblock {\em New J. Phys.}, 12(5):075008, 2010.

\bibitem{DongPetersen2011}
D.~Dong and I.~Petersen.
\newblock Quantum control theory and applications: a survey.
\newblock {\em IET Control theory and applications}, 4(12):2651 --2671, 2011.

\bibitem{AltafiniTicozzi2012}
C.~Altafini and F.~Ticozzi.
\newblock Modeling and control of quantum systems: an introduction.
\newblock {\em IEEE Transactions on Automatic Control}, 57:1898 -- 1917, 2012.

\bibitem{Lindblad76}
G.~Lindblad.
\newblock On the generators of quantum dynamical semigroups.
\newblock {\em Comm. Math. Phys.}, 48:119, 1976.

\bibitem{GoriniKossakowskiSudarshan76}
V.~Gorini, A.~Kossakowski, and E.C.G. Sudarshan.
\newblock Completely positive dynamical semigroups of ${N}$-level systems.
\newblock {\em J. Math. Phys.}, 17(5):821, 1976.

\bibitem{BreuerPetruccioneBook}
H.-P. Breuer and F.~Petruccione.
\newblock {\em The Theory of Open Quantum Systems}.
\newblock Oxford University Press, 2007.

\bibitem{LloydViola01}
S.~Lloyd and L.~Viola.
\newblock Engineering quantum dynamics.
\newblock {\em Phys. Rev. A}, 65:010101, 2001.

\bibitem{Baconetal01}
D.~Bacon et~al.
\newblock Universal simulation of {M}arkovian quantum dynamics.
\newblock {\em Phys. Rev. A}, 64:062302, 2001.

\bibitem{Barreiroetal2011}
J.T. Barreiro et~al.
\newblock An open-system quantum simulator with trapped ions.
\newblock {\em Nature}, 470:486, 2011.

\bibitem{TannorBartana99}
D.~J. Tannor and A.~Bartana.
\newblock On the interplay of control fields and spontaneous emission in laser
  cooling.
\newblock {\em J. Phys. Chem. A}, 103:10359, 1999.

\bibitem{SklarzTannorKhaneja04}
S.~E. Sklarz, D.~J. Tannor, and N.~Khaneja.
\newblock Optimal control of quantum dissipative dynamics: Analytic solution
  for cooling the three-level ${\Lambda}$ system.
\newblock {\em Phys. Rev. A}, 69:053408, 2004.

\bibitem{Schirmeretal04}
S.~G. Schirmer, T.~Zhang, and J.V. Leahy.
\newblock Orbits of quantum states and geometry of {B}loch vectors for
  ${N}$-level systems.
\newblock {\em J. Phys. A}, 37:1389, 2004.

\bibitem{Khanejaetal02}
N.~Khaneja, S.J. Glaser, and R.W. Brockett.
\newblock Sub-riemannian geometry and time optimal control of three spin
  systems: Quantum gates and coherence transfer.
\newblock {\em Phys. Rev. A}, 65:032301, 2002.

\bibitem{SchirmerWang2010}
S.~Schirmer and X.~Wang.
\newblock Stabilizing open quantum systems by markovian reservoir engineering.
\newblock {\em Physical Review A}, 81:062306, 2010.

\bibitem{us_nis2_a}
P.~Rooney, A.M. Bloch, and C.~Rangan.
\newblock Decoherence control and purification of two-dimensional quantum
  density matrices under {L}indblad dissipation.
\newblock 2012.
\newblock arXiv:1201.0399v1 [quant-ph].

\bibitem{Yuan2013}
H.~Yuan.
\newblock Reachable set of open quantum dynamics for a single spin in markovian
  environment.
\newblock {\em Automatica}, 49:955--959, 2013.

\bibitem{StewartBook}
J.~Stewart.
\newblock {\em Multivariate Calculus, 8th ed.}
\newblock Brooks Cole, 2015.

\bibitem{LangBook}
S.~Lang.
\newblock {\em Introduction to Differential Manifolds, 2nd ed.}
\newblock Springer-Verlag, 2002.

\bibitem{Preskill}
J.~Preskill.
\newblock Lecture notes on quantum computation.
\newblock
  \url{http://www.theory.caltech.edu/people/preskill/ph229/notes/chap3.pdf}.

\bibitem{Celledonietal2012}
E.~Celledoni, H.~Marthinsen, and B.~Owren.
\newblock An introduction to {L}ie group integrators - basics, new developments
  and applications.
\newblock {\em J. Comp. Phys.}, 257:1040 -- 1061, 2014.

\bibitem{Schur}
I.~Schur.
\newblock \"{U}ber eine {K}lasse von {M}ittelbildungen mit {A}nwendungen auf
  die {D}eterminantentheorie.
\newblock {\em Sitzungsber. Berl. Math. Ges.}, 22:9, 1923.

\bibitem{Horn}
A.~Horn.
\newblock Doubly stochastic matrices and the diagonal of a rotation matrix.
\newblock {\em Am. J. Math.}, 76:620, 1954.

\end{thebibliography}
\end{document}